\documentclass[11pt,a4paper, english]{article}
\usepackage[utf8]{inputenc}			
\usepackage{makeidx}
\usepackage[pdftex]{graphicx}
\usepackage{mathtools,amsfonts}
\usepackage{array}
\usepackage{multirow, hhline}
\usepackage{paralist}

\usepackage[usenames,dvipsnames,table]{xcolor}
\usepackage{wrapfig}
\usepackage[english]{babel}		
\usepackage[margin=1in]{geometry}	
\usepackage{lipsum}			
\usepackage{comment}
\usepackage{chngpage}			
\usepackage{amsfonts,amssymb,amsmath, amsthm}
\usepackage{amsopn}
\usepackage{color}

\usepackage{pgfplots}
\usepackage{pgfgantt}
\usepackage{tikz}
\usetikzlibrary{arrows}
\usetikzlibrary{calc}
\usepackage{algpseudocode}
\usepackage[Algorithm]{algorithm}
\usepackage{subcaption}
\usepackage{tabto}
\usepackage{csquotes} 

\newtheorem{theorem}{Theorem}

\newtheorem{definition}[theorem]{Definition}

\newtheorem{lemma}[theorem]{Lemma}

\newtheorem{proposition}[theorem]{Proposition}

\definecolor{green}{RGB}{10,180,10} 
\usepackage[breaklinks=true]{hyperref} 
\hypersetup{
	colorlinks=true,
	citecolor=green,
	linkcolor=blue,  
}

\newcommand{\poly}{\textup{poly}}

\newcommand{\MMS}{\mathit{MMS}}

\newcommand{\wM}{\widetilde{M}}

\title{A Little Charity Guarantees Almost Envy-Freeness}

\author{\begin{tabular}{ccc}Bhaskar Ray Chaudhury\thanks{MPI for Informatics, Saarland Informatics Campus, Germany. {\tt $\{$braycha, mehlhorn$\}$@mpi-inf.mpg.de}} &
Telikepalli Kavitha\thanks{Tata Institute of Fundamental Research, India. Work done at MPI for Informatics, Saarland Informatics Campus, Germany. {\tt kavitha@tifr.res.in}} &
           \setcounter{footnote}{0} Kurt Mehlhorn\footnotemark \\ \setcounter{footnote}{2}
          & Alkmini Sgouritsa\thanks{University of Liverpool, UK. Work done at MPI for Informatics, Saarland Informatics Campus, Germany. {\tt alkmini@liv.ac.uk}}
 \end{tabular}}

\newcommand{\abs}[1]{| #1 |}

\newcommand{\sset}[1]{\{ #1 \} }

\date{}

\begin{document}
\maketitle

\begin{abstract}
  
Fair division of indivisible goods is a very well-studied problem. The goal of this problem is to distribute $m$ goods to $n$ agents in a ``fair'' manner, where every agent has a valuation for each subset of goods. We assume general valuations.

Envy-freeness is the most extensively studied notion of fairness. However, envy-free allocations do not always exist when goods are indivisible. The notion of fairness we consider here is ``envy-freeness up to any good'' (EFX) where no agent envies another agent after the removal of any single good from the other agent's bundle. It is not known if such an allocation always exists. 

We show there is always a partition of the set of goods into $n+1$ subsets $(X_1,\ldots,X_n,P)$ where for $i \in [n]$,
$X_i$ is the bundle allocated to agent $i$ and the set $P$ is unallocated (or donated to charity) such that we have:
\begin{itemize}
\item envy-freeness up to any good,
\item no agent values $P$ higher than her own bundle, and
\item fewer than $n$ goods go to charity, i.e., $|P| < n$ (typically $m \gg n$).
\end{itemize}
Our proof is constructive and leads to a pseudo-polynomial time algorithm to find such an allocation. When agents have additive valuations and $\abs{P}$ is large (i.e., when $|P|$ is close to $n$), our allocation also has a good maximin share (MMS) guarantee. Moreover, a minor variant of our algorithm also shows the existence of an allocation that is $4/7$ \emph{groupwise} maximin share (GMMS): this is a notion of fairness stronger than MMS. This improves upon the current best bound of $1/2$ known for an approximate GMMS allocation.

\vspace{0.5cm}

\footnotesize{ This version of the paper goes beyond the preliminary version published in SODA 2020~\cite{CKMS20} in two key points:
\begin{enumerate}
    \item The pseudo-polynomial algorithm works when agents have general valuation functions (not just gross-substitute valuations).
     \item We introduce a relaxed definition of ``most envious agent''.
\end{enumerate}
}

\end{abstract}

\section{Introduction}
\label{sec:intro}
Fair division of goods among competing agents is a fundamental problem in Economics and Computer Science. There is a set $M$ of $m$ goods and the goal is to {\em allocate} goods among $n$ agents in a {\em fair} way. An allocation is a partition of $M$ into disjoint subsets $X_1,\ldots,X_n$ where $X_i$ is the set of goods given to agent $i$.
When can an allocation be considered ``fair''? One of the most well-studied notions of fairness is {\em Envy-freeness}. Every agent has a {\em value} associated with each subset of $M$ and agent~$i$ {\em envies} agent $j$ if $i$ values $X_j$ more than $X_i$. An allocation is {\em envy-free} if no agent envies another. An envy-free allocation can be regarded as a fair and desirable partition of $M$ among the $n$ agents since no agent envies another; as mentioned in~\cite{TimPlaut18}, such a mechanism of partitioning land dates back to the Bible.

Unlike land which is divisible, goods in our setting are {\em indivisible} and an envy-free allocation of the given set of goods need not exist. Consider the following simple example with two agents and a single good that both agents desire: one of the agents has to receive this good and the other agent envies her. Since envy-free allocations need not exist, several relaxations have been considered.

\medskip

\noindent{\bf Relaxations.}
Budish \cite{budish2011combinatorial} introduced the notion of {\em EF1}: this is an allocation of goods that is ``envy-free up to \emph{one} good''. In an EF1 allocation, agent $i$ may envy agent $j$, however this envy would vanish as soon as {\em some} good is removed from $X_j$. Note that no good is really removed from $X_j$: this is simply a way of assessing how much $i$ values $X_j$ more than $X_i$. That is, if $i$ values $X_j$ more than $X_i$, then there exists some $g \in X_j$ such that $i$ values $X_i$ at least as much as $X_j \setminus \{g\}$. Going back to the example of two agents and a single good, the allocation where one agent receives this good is EF1. It is known that EF1 allocations always exist; as shown by Lipton et al.\ \cite{LiptonMMS04}, such an allocation can be efficiently computed.\footnote{The algorithm in \cite{LiptonMMS04} was published in 2004 with a different property and EF1 was proposed in 2011.}

Caragiannis et al.\ \cite{CaragiannisKMP016} introduced a notion of envy-freeness called {\em EFX} that is stronger than EF1. An EFX allocation is one that is ``envy-free up to \emph{any} good''. In an EFX allocation, agent~$i$ may envy agent~$j$, however this envy would vanish as soon as {\em any} good is removed from $X_j$. Thus every EFX allocation is also EF1 but not every EF1 allocation is EFX. 
\begin{center}
\begin{minipage}[b]{0.3\linewidth}
\centering
\begin{eqnarray*}
  \setlength{\arraycolsep}{0.5ex}\setlength{\extrarowheight}{0.25ex}
\begin{array}{@{\hspace{1ex}}c@{\hspace{1ex}}||@{\hspace{1ex}}c@{\hspace{1ex}}|@{\hspace{1ex}}c@{\hspace{1ex}}|@{\hspace{1ex}}c@{\hspace{1ex}}|@{\hspace{1ex}}c@{\hspace{1ex}}}
    \  & a \ & b \ & c \ \\[.5ex] \hline
    Agent~1 \ & 1 \ & 1 \ & 2 \ \\[.5ex] \hline
    Agent~2 \ & 1 \ & 1 \ & 2 \ \\[.5ex] 
\end{array}
\end{eqnarray*}
\end{minipage}
\end{center}

Consider the simple example given above: there are three goods $a, b, c$ and two agents with {\em additive valuations}\footnote{The value of a set $S \subseteq M$ is the sum of values of goods in $S$.} (formally defined in Section~\ref{sec:results}) as described below.
Both agents value $c$ twice as much as $a$ or $b$. The allocation where agent~1 gets $\{a\}$ and agent~2 gets $\{b,c\}$ is  EF1 but not EFX. On the other hand, the allocation where agent~1 gets $\{a,b\}$ and agent~2 gets $\{c\}$ is EFX. Indeed, the latter allocation seems fairer than the former allocation.
As said in \cite{CaragiannisGravin19}: ``{\em Arguably, EFX is the best fairness analog of envy-freeness of indivisible items}''.
While it is known that EF1 allocations always exist, the question of whether EFX allocations always exist is still an open problem (despite significant effort, according to \cite{CaragiannisKMP016}). 

Plaut and Roughgarden \cite{TimPlaut18} showed that EFX allocations always exist (i)~when there are only two agents or (ii) when all $n$ agents have the same valuations.
Moreover, it was shown in \cite{TimPlaut18} that exponentially many value queries may be needed to determine EFX allocations even in the restricted setting where there are only two agents with identical {\em submodular} valuation functions\footnote{These are  valuation functions with decreasing marginal values.}. 
It was not known until very recently, if an EFX allocation always exists even when there are only three agents with additive valuations. A positive answer to this question was very recently given by Chaudhury et al.\ \cite{CGM20} who showed that EFX allocations always exist when there are only three agents with additive valuations. It was remarked in \cite{TimPlaut18}: ``{\em We suspect that at least for general valuations, there exist instances where no EFX allocation exists}''. 

\medskip

\noindent{\bf A relaxation of EFX.}
Very recently, Caragiannis et al.\ \cite{CaragiannisGravin19} introduced a more relaxed notion of EFX called {\em EFX-with-charity}. This is a partial allocation that is EFX, i.e., the entire set of goods need not be distributed among the agents. So some goods may be left {\em unallocated} and it is assumed that these unallocated goods are donated to charity. There is a very simple allocation that is EFX-with-charity where {\em no} good is assigned to any agent---thus all goods are donated to charity. Obviously, this is not an interesting allocation and one seeks allocations with better guarantees and one such allocation was shown in \cite{CaragiannisGravin19}. 

Let $X^* = \langle X^*_1,\ldots,X^*_n\rangle$ be an  optimal {\em Nash social welfare} allocation\footnote{This is an allocation that maximizes $\Pi_{i=1}^nv_i(X^*_i)$, where $v_i$ is agent $i$'s valuation function.} on the entire set of goods. It was shown in \cite{CaragiannisGravin19} that 
there always exists an EFX-with-charity allocation $X = \langle X_1,\ldots,X_n\rangle$ where every agent receives at least half the value of her allocation in $X^*$. Interestingly, as shown in \cite{CaragiannisGravin19}, $X_i \subseteq X^*_i$ for all $i$. Unfortunately, there are no upper bounds on 
the {\em number} of unallocated goods or on the {\em value} any agent has for the set of goods donated to charity.

We believe these are important questions to ask. The ideal allocation is one that is EFX and 
allocates {\em all} goods; so we would like a guarantee that a large number of goods have been allocated to agents. Moreover, since EFX allocations guarantee envy-freeness once any good is removed from another agent's set, it is in the same spirit that we seek an EFX (partial) allocation where nobody envies the set of unallocated goods. 
The allocation in \cite{CaragiannisGravin19} gives no guarantee either on the number of  unallocated goods or on whether any agent values the set of unallocated goods more than her own bundle. 
Here we consider the notion of {\em EFX-with-bounded-charity}. That is, we seek EFX-with-charity allocations with bounds on the set given to charity, i.e., a bound on the {\em size} and a bound on the {\em value} of the set of goods donated to charity.

\subsection{Our Results}
\label{sec:results}
Let $N = [n]$ be the set of agents.
Every agent $i \in [n]$ has a valuation function $v_i: 2^M \rightarrow \mathbb{R}_{\ge 0}$, where $M$ is the set of $m$ goods. 

\subsubsection{General valuations}
We show our main existence result for \emph{general} valuation functions, i.e., the only assumptions we make on any valuation function $v_i$ are that
\begin{itemize}
\item[(i)] it is {\em normalized}, i.e., $v_i(\emptyset) = 0$, and (ii)~it is {\em monotone}, i.e., $S \subseteq T$ implies $v_i(S) \le v_i(T)$.
\end{itemize}
In contrast, the EFX-with-charity allocation in  \cite{CaragiannisGravin19} works only for {\em additive valuations}, i.e., for any $S \subseteq M$ and $i \in [n]$, we have $v_i(S) = \sum_{g\in S} v_i(\left\{g\right\})$.

\smallskip

\begin{itemize}
\item We show there always exists an allocation\footnote{Henceforth, allocations are {\em partial} and we will use ``complete allocation'' to refer to one where all goods are allocated.} $X = \langle X_1,\ldots,X_n\rangle$ that satisfies the following properties:
\begin{enumerate}
\item $X$ is EFX, i.e., for any two agents $i, j$:  $v_i(X_i) \ge v_i(X_j \setminus \{g\})$ for any $g \in X_j$;
\item $v_i(X_i) \ge v_i(P)$ for all agents $i$, where 
  $P = M \setminus \cup_{i=1}^n X_i$ is the set of unallocated goods; 
\item $|P| < n$ (recall that $n$ is the number of agents).
\end{enumerate}
\end{itemize}

Note that our result implies that among the $n$ agents, if there is just {\em one} agent who has no preferences (say, for agent~$i$, we have $v_i(S) = 0$ for all $S \subseteq M$) then a complete EFX allocation always exists for general valuations: Find an EFX allocation satisfying all the three properties mentioned above, on the set of agents $N \setminus \left\{i\right\}$. Note that the allocation among the agents in $N \setminus \left\{i\right\}$ is EFX (by condition 1) and we now allocate $P$ (the set of unallocated goods) to agent $i$. Observe that agent~$i$ envies nobody as $i$ has the same valuation for all subsets of goods and nobody envies $i$ as nobody envies $P$ (by condition 2). Thus we have a complete EFX allocation! 

Our proof is constructive. We start with no goods being allocated to the agents and find the claimed allocation by at most $nmV/\Delta$ applications of three simple update rules, where $n$ is the number of agents, $m$ is the number of goods, $V = \max_i v_i(M)$ is the maximum valuation of any agent, and $\Delta = \min_i \min \{|v_i(T) - v_i(S)|: S,T \subseteq M \text{ and } v_i(S) \ne v_i(T)\}$ is the minimum difference between distinct valuations.

The update rules use a \emph{minimal-envied-subset}-oracle: given $S \subseteq M$ such that there is an agent who values $S$ more than her own bundle, find an inclusion-wise minimal subset $Z \subseteq S$ such that there is an  agent who values $Z$ more than her own bundle. This oracle can be realized by a simple algorithm that uses at most $nm$ value queries. 

\begin{itemize}
\item  For general valuations, an EFX allocation with properties~1-3 can be computed with $\poly(n,m,V,1/\Delta)$ value queries, i.e., in
  pseudo-polynomial time.
\end{itemize}


\subsubsection{Identical valuations}
It also follows from our proof that when all agents have the same valuation function, our allocation is {\em complete}. That is, $|P| = 0$. This is an alternate proof to the existence of complete EFX allocations for identical (general) valuations, originally shown in \cite{TimPlaut18}.


\subsubsection{Additive valuations}
The most well-understood class of valuation functions is the set of {\em additive} valuations.
We consider the case when all agents have additive valuations and show that our allocation or very minor variants of our allocation can guarantee several other notions of fairness.

\medskip

\noindent{\bf Ensuring high Nash social welfare.}
We show that modifying the starting step of our algorithm ensures that our allocation $X$, that satisfies properties~1-3 stated above, also has a high Nash social welfare. That is, $v_i(X_i) \ge \frac{1}{2}\cdot v_i(X^*_i)$ as promised in \cite{CaragiannisGravin19}, where $X^* = \langle X^*_1,\ldots,X^*_n\rangle$ is an  optimal Nash social welfare allocation. Here we use the allocation computed in \cite{CaragiannisGravin19} as a black box in our starting step, thus our result can be regarded as an extension of the result in \cite{CaragiannisGravin19}.

\medskip

\noindent{\bf Number of Unallocated Goods and MMS Guarantee.}
Another interesting and well-studied notion of fairness is \emph{maximin share}. Suppose agent $i$ has to partition $M$ into $n$ {\em bundles} (or sets) knowing that she would receive the worst bundle with respect to her valuation. Then $i$ will choose a partition of $M$ that maximizes the valuation of the worst bundle (wrt her valuation). The value of this worst bundle is the maximin share of agent~$i$. An important question here is: does there always exist an allocation of $M$ where every agent gets a bundle worth at least her maximin share?

Formally, let $N$ and $M$ be the sets of $n$ agents and $m$ goods, respectively. We define the maximin share of an agent (say, $i$) as follows:
(here $\mathcal{X}$ is the set of all complete allocations)
\[\MMS_i(n,M) = \max_{\langle X_1,\dots,X_n \rangle \in \mathcal{X}} \min_{j \in [n]} v_i(X_j).\]

The goal is to determine an allocation $\langle X_1,X_2, \dots ,X_n \rangle$ of $M$ such that for every $i$ we have $v_i(X_i) \geq \MMS_i(n,M)$. This question was first posed by Budish \cite{budish2011combinatorial}. Procaccia and Wang~\cite{ProcacciaW14} showed that such an allocation need not exist, even in the restricted setting of only three agents! 
Thereafter, {\em approximate-MMS} allocations were studied~\cite{ProcacciaW14, JGargMT19, GhodsiHSSY18, GargTaki19} and there are polynomial time algorithms to find allocations where for all $i$, agent~$i$ gets a bundle of value at least $\alpha\cdot\MMS_i(n,M)$; the current best guarantee for $\alpha$ is $3/4-\epsilon$ by Ghodsi et al.\ \cite{GhodsiHSSY18} (for any $\epsilon > 0$) and this was very recently improved to $3/4$ by Garg and Taki~\cite{GargTaki19}.

Amanatidis et al.\ \cite{ABM18} showed that any complete EFX allocation is also a $\tfrac{4}{7}$-MMS allocation. We show that our allocation promises better MMS guarantees when the number of unallocated goods is large.
Let $X = \langle X_1,\ldots,X_n\rangle$ be our allocation as described by properties~1-3 above and $P$ be the set of 
unallocated goods. For any agent $i \in [n]$, we have:
\[v_i(X_i) \ \ge\ \frac{1}{2 - |P|/n}\MMS_i(n,M).\]

Hence, larger the number of unallocated goods, better are the guarantees that we get on MMS. The extreme values are $|P|= 0$ and $|P| = n-1$. When
$|P| = 0$, we have a complete EFX allocation and when $|P| = n-1$, we have an EFX allocation that is an {\em almost-MMS} allocation: 
$v_i(X_i) \ge (1-1/n)\cdot\MMS_i(n,M)$ for all $i$.

\medskip

\noindent{\bf Improved Guarantees for Groupwise MMS.}
Barman et al.\ \cite{BarmanBMN18} recently introduced a notion of fairness called {\em groupwise maximin share} (GMMS) which is stronger than MMS. An allocation is said to be GMMS if the MMS condition is satisfied for {\em every subgroup} of agents and the union of the sets of goods allocated to them. Formally, 
a complete allocation $X = \langle X_1,X_2, \dots , X_n \rangle$ is $\alpha$-GMMS if for any $N' \subseteq N$, we have $v_i(X_i) \geq \alpha\cdot\MMS_i(n', \bigcup_{i \in N'}X_i)$ where $n' = \abs{N'}$. 
Every GMMS allocation, i.e. $\alpha = 1$, is also a complete EFX allocation~\cite{BarmanBMN18}.

It is known~\cite{BarmanBMN18} that GMMS strictly generalizes MMS. In particular, it was shown in \cite{BarmanBMN18} that GMMS allocations rule out some very unsatisfactory allocations that have MMS guarantees. For example, consider an instance with $n$ agents with additive valuations and  a set $M$ of $n-1$ goods and every agent has a valuation of one for each good. Since the number of goods is less than the number of agents, we have $\MMS_i(n,M) = 0$ for every agent $i$. So any allocation has MMS guarantees. It is not hard to see that the only allocation with a GMMS guarantee is one where $n-1$ agents get one good each and one agent is left without any goods. See Subsection~2.1 in~\cite{BarmanBMN18} for more discussion.
Naturally, it is a harder problem to approximate GMMS than MMS. While $\frac{3}{4}$-MMS allocations always exist, the largest $\alpha$ for which $\alpha$-GMMS allocations are known to exist is $\frac{1}{2}$~\cite{BarmanBMN18}.
We extend the result of Amanatidis et al.~\cite{ABM18} for MMS to show the following:

\begin{itemize}
\item A $\frac{4}{7}$-GMMS allocation always exists and can be computed in pseudo-polynomial time.
\end{itemize}
  In particular, we show that modifying the last step of our algorithm results in a complete allocation that is $\tfrac{4}{7}$-GMMS. Very recently and independently, Amanatidis et al.\ \cite{ANM19} showed the same approximation.

\subsection{Our Techniques}
\label{sec:techniques}


We now give an overview of the main ideas used to find our EFX allocation. We first recall the algorithm of Lipton et al.\ \cite{LiptonMMS04} for finding an EF1 allocation. 
They use the notion of an {\em envy-graph}: here each vertex corresponds to an agent and there is an edge $(i,j)$ iff $i$ envies $j$. The invariant maintained is that the envy-graph is a DAG: a cycle corresponds to a cycle of envy and by swapping bundles along a cycle, every agent becomes better-off and the number of envy edges 
does not increase. More precisely, if $i_0 \rightarrow i_1 \rightarrow i_2 \rightarrow \ldots \rightarrow  i_{\ell - 1} \rightarrow i_0$ is a cycle in the envy graph, then reassigning $X_{i_{j+1}}$ to agent $i_j$ for $0 \le j < \ell$ (indices are to be read modulo $\ell$) will increase the valuation of every agent in the cycle. Also if there was an edge from $s$ to some $i_k$ where $s$ is not a part of the cycle, this edge just gets directed now from $s$ to $i_{k+1}$ after we exchange bundles along the cycle. Thus the number of envy edges in the graph does not increase and the valuations of the agents in the cycle goes up. Thus cycles can be eliminated.

The algorithm in \cite{LiptonMMS04} runs in rounds and always maintains an allocation that is also EF1. At the beginning of every round, an unenvied agent $s$ (this is a source vertex in this DAG) is identified and an unallocated good $g$ is allocated to $s$.  The new allocation is also EF1, as nobody will envy the bundle of $s$ after removing the good $g$. 

\medskip

\noindent{\bf The Reallocation Operation.}
We now highlight a key difference between an EF1 allocation and an EFX allocation. From the algorithm of Lipton et al.\ \cite{LiptonMMS04}, it is clear that given an EF1 allocation on a set $M_0$ of goods, one can determine an EF1 allocation on $M_0 \cup M_1$, for any $M_1 \subseteq M \setminus M_0$, by simply adding goods from $M_1$ one-by-one to the existing bundles and changing the owners (if necessary) in a clever way. Intuitively, we never need to cut or merge the bundles formed in any EF1 allocation. We can just append the unallocated goods appropriately to the current bundles.

The above strategy is very far from true for EFX. Consider the example illustrated below with three agents with additive valuations and four goods $a, b$, $c$, and $d$. 

\begin{center}
\begin{minipage}[b]{0.3\linewidth}
\centering
\begin{eqnarray*}
  \setlength{\arraycolsep}{0.5ex}\setlength{\extrarowheight}{0.25ex}
\begin{array}{@{\hspace{1ex}}c@{\hspace{1ex}}||@{\hspace{1ex}}c@{\hspace{1ex}}|@{\hspace{1ex}}c@{\hspace{1ex}}|@{\hspace{1ex}}c@{\hspace{1ex}}|@{\hspace{1ex}}c@{\hspace{1ex}}|@{\hspace{1ex}}c@{\hspace{1ex}}}
    \  & a \ & b \ & c \  & d \ \\[.5ex] \hline
    Agent~1 \ & 0 \ & 1 \ & 1 \ & 2 \ \\[.5ex] \hline
    Agent~2 \ & 1 \ & 0 \ & 1 \ & 2 \ \\[.5ex] \hline
    Agent~3 \ & 1 \ & 1 \ & 0 \ & 2 \ \\[.5ex]
\end{array}
\end{eqnarray*}
\end{minipage}
\end{center}

An EFX allocation for the first three goods has to give exactly one of $a, b, c$ to each of the three agents. However an EFX allocation for all the four goods has to allocate the singleton set $\{d\}$ to some agent (say, agent~1) and say, $\{a\}$ to agent~2 and $\{b,c\}$ to agent~3. Thus the allocation needs to be {\em cut and merged}. When there are many agents - each with her own valuation, figuring out the cut-and-merge operations is the difficult step. 
Here we implement our global reallocation operation as follows.

\medskip

\noindent{\bf Improving Social Welfare.}
Suppose we have an EFX allocation $X = \langle X_1,\ldots,X_n\rangle$ on some subset $M_0 \subset M$. We would now like to add a good $g \in M \setminus M_0$. However we will not be able to guarantee an EFX allocation on $M_0 \cup \{g\}$. What we will ensure is that either case~(i) or case~(ii) occurs:  
\begin{itemize}
\item[(i)] We have an EFX allocation $X' = \langle X'_1,\ldots,X'_n\rangle$ on a subset of $M_0 \cup \{g\}$ such that $v_i(X'_i) \ge v_i(X_i)$ for all $i$ and for at least one agent $j$ we have $v_j(X'_j) > v_j(X_j)$. Thus $\sum_{i\in[n]} v_i(X'_i) > \sum_{i\in[n]} v_i(X_i)$; in other words, the {\em social welfare} strictly improves. 
\item[(ii)] We have an EFX allocation on $M_0 \cup \{g\}$ and the social welfare does not decrease.
\end{itemize}

Hence in each step of our algorithm, we either increase social welfare or we increase the number of allocated goods without decreasing social welfare---thus we always make progress. This is similar to the approach used by Plaut and Roughgarden~\cite{TimPlaut18} 
to guarantee the existence of $\tfrac{1}{2}$-EFX\footnote{An allocation $X = (X_1,\ldots,X_n)$ is $\tfrac{1}{2}$-EFX if for any two agents $i, j$: $v_i(X_i) \ge \frac{1}{2}\cdot v_i(X_j\setminus\{g\})$ for all $g \in X_j$.} when agents have subadditive valuations. We now outline how we ensure that one of cases (i) and (ii) has to happen.

For simplicity of exposition, we assume the envy-graph corresponding to our starting EFX allocation $X$ has a single source $s$. Add $g$ to $s$'s bundle: if nobody envies $s$ up to any good then we are in an easy case as we have an EFX allocation on $M_0 \cup \{g\}$. In this case, we ``decycle'' the envy-graph (if cycles are created) and continue. Observe that swapping bundles along a cycle in the envy-graph increases social welfare.

\medskip

\noindent{\bf Most Envious Agent.}
Assume now that there are one or more agents who envy $s$ up to any good after $g$ is allocated to $s$. To resolve this, we introduce the concept of the \emph{most envious agent}. Let $i$ be an agent who envies $s$ up to any good, so $v_i(X_i) < v_i(S')$ for some $S' \subset X_s\cup\{g\}$. Let $S_i$ be any inclusion-wise minimal subset of $X_s\cup\{g\}$ such that $v_i(X_i) < v_i(S_i)$ (break ties arbitrarily). So for any $T \subset S_i$, we have $v_i(X_i) \ge v_i(T)$. 

\begin{itemize}
\item An agent $i$ such that $v_i(X_i) < v_i(S_i)$ for some  $S_i \subset X_s\cup\{g\}$ and no strict subset of $S_i$ is envied by {\em any} agent will be called the \emph{most envious agent} of $X_s\cup\{g\}$ (break ties arbitrarily). 
\end{itemize}

Let $t$ be the most envious agent of $X_s\cup\{g\}$.
The crucial observation is that no agent envies $S_t$ {\em up to any good}---otherwise it would contradict $S_t \subset X_s\cup\{g\}$ being an ``inclusion-wise minimal envied set'', i.e., no proper subset of $S_t$ is envied by any agent. Recall the assumption that $s$ is the only source, so there is a path $s \rightarrow i_1 \rightarrow \cdots \rightarrow i_{k-1} \rightarrow t$ in the envy-graph. We do a leftwise shift of bundles along this path: so $s$ gets $i_1$'s bundle, and for $1 \le r \le k-1$: $i_r$ gets $i_{r+1}$'s bundle (where $i_k = t$), and finally $t$ gets $S_t$. The goods in $X_s \cup \{g\} \setminus S_t$ are thrown back into the pool of unallocated goods.

Observe that every agent in this path is strictly better-off now than in the allocation $X$ and no other agent is worse-off. Moreover, by the definition of $S_t$, there is no agent envying any other agent up to any good. Thus we have a desired EFX allocation $X'$. When there are multiple sources, we can adapt this technique {\em provided} there are enough unallocated goods; in particular, the number of unallocated goods must be at least the number of sources in the envy-graph. We describe this in detail in Section~\ref{sec:main}.

We would like to contrast the above approach with other EFX algorithms~\cite{TimPlaut18,CaragiannisGravin19}. The $\tfrac{1}{2}$-EFX algorithm by Plaut and Roughgarden \cite{TimPlaut18} either merges $g$ (the new good) with an existing bundle or allocates the singleton set $\{g\}$ to an agent. The EFX-with-charity algorithm by Caragiannis et al.\ \cite{CaragiannisGravin19} takes an allocation of maximum Nash social welfare as input and then permanently removes some goods from the instance. We regard the notion of ``most envious agent'' that shows a
natural way of breaking up a bundle to preserve envy-freeness up to any good as one of the innovative contributions of our work.

\medskip

\noindent{\bf Our Other Results.}
Regarding our result with approximate MMS guarantee, if the number of unallocated goods in our EFX allocation is large, then the number of sources also has to be large: these are unenvied agents. Moreover, no agent envies the set of unallocated goods. Suppose for now that $|P| = n-1$. This means every agent is a source. So no agent envies the bundle of any other agent and also the set of unallocated goods. Thus for each agent $i$, we have:
\begin{eqnarray*}
  v_i(X_i) \ \geq \ \frac{v_i(M)}{n+1} \ & \geq & \ (1 + 1/n)^{-1} \cdot \frac{v_i(M)}{n} \\
  & \ge  & \ (1 - 1/n)\cdot\MMS_i(n,M),
\end{eqnarray*}  
where the constraint that $v_i(M)/n \ge \MMS_i(n,M)$ holds for additive valuations. We show our result for approximate-MMS allocation and our improved bound for approximate-GMMS allocation in Section~\ref{sec:additive-valuations}.

\subsection{Related Work}
Fair division of divisible resources is a classical and well-studied subject starting from the 1940's~\cite{Steinhaus48}. Fair division of indivisible goods among competing agents is a young and exciting topic with recent work on EF1 and EFX allocations~\cite{CaragiannisKMP016,BKV18,TimPlaut18,BCFIMPVZ19,CaragiannisGravin19}, approximate maximin share allocations~\cite{budish2011combinatorial,BL16,AMNS17,BK17,KPW18,GhodsiHSSY18,JGargMT19}, and approximation algorithms for maximizing Nash social welfare and generalizations~\cite{CG15,CDGJMVY17,ChaudhuryCG0HM18,AOSS17,GHM18,AMOV18}. 

As mentioned earlier, Caragiannis et al.\ \cite{CaragiannisKMP016} introduced the notion of EFX and it is now known that EFX allocations always exist for three agents with additive valuations~\cite{CGM20}. Whether EFX allocations always exist with general valuations or with a larger number of agents is an enigmatic open problem.  It was shown in \cite{CaragiannisKMP016} that there always exists an EF1 allocation that is also Pareto-optimal\footnote{An allocation $X = \langle X_1,\ldots,X_n\rangle$ is Pareto-optimal if there is no allocation $Y = \langle Y_1,\ldots,Y_n\rangle$ where $v_i(Y_i) \ge v_i(X_i)$ for all $i \in [n]$ and $v_i(Y_j) > v_i(X_j)$ for some $j$.} 
and Barman et al.\ in \cite{BKV18} showed a pseudo-polynomial time algorithm to compute such an allocation.

\smallskip

\noindent{\bf Applications.}
Fair division of goods or resources occurs in many real-world scenarios and this is demonstrated by the popularity of the website Spliddit (\url{http://www.spliddit.org}) that implements mechanisms for fair division where users can log in, define what needs to be divided, and enter their valuations. This website guarantees an EF1 allocation that is also Pareto-optimal and since its launch in 2014, it has been used tens of thousands of times~\cite{CaragiannisKMP016}. We refer to \cite{GP14,TimPlaut18} for details on the diverse applications for which Spliddit has been used: these range from rent division and taxi fare division to credit assignment for an academic paper or group project. Another such website is Fair Outcomes, Inc. (\url{http://www.fairoutcomes.com}). An interesting application is also {\em Course Allocate} used at Wharton School that guarantees certain fairness properties to allocate courses among students~\cite{TimPlaut18}.

\subsection{Improvements with respect to the SODA version}
In the conference version of the paper~\cite{CKMS20}, we gave a pseudo-polynomial time algorithm to determine an EFX allocation with bounded charity when agents have ``gross-substitute'' valuations. In this version we show a pseudo-polynomial time algorithm to determine an EFX allocation with bounded charity even when agents have general valuations. This is realized with a relaxed definition of the ``most envious agent'' (see Definition~\ref{mostenviousagent-definition}) that helps us to implement the \emph{minimal-envied-subset} oracle efficiently (see Subsection~\ref{minimum-envied-subset-finding}).

\section{Existence of an EFX-Allocation with Bounded Charity}
\label{sec:main}

We prove our main result on EFX-with-bounded-charity allocations in this section. We will define three {\em update rules}. Each update rule takes a pair $(X,P)$ consisting of an allocation $X$ and a set $P$ of unallocated goods (we will call $P$ the {\em pool}) and returns a modified pair $(X',P')$. 

Each application of an update rule will ensure that either (i)~the \emph{social welfare} $\phi(X) = \sum_{i \in [n]}v_i(X_i)$ of the current allocation increases or (ii)~the size of the pool decreases and the social welfare 
does not decrease, so  $|P'| < |P|$ in this case.
Hence the update process will terminate. The overall structure of the algorithm is given in Algorithm~\ref{main algorithm}.

\algrenewcommand{\algorithmiccomment}[1]{$\Rightarrow$ #1}
  
\begin{algorithm}[h]
  \begin{algorithmic}[1]
    \Statex {\textbf{Postcondition}: \parbox[t]{0.8\textwidth}{$X$ is EFX, $\abs{P} < n$ and $v_i(P) \le v_i(X_i)$ for all $i \in [n]$.}}
    \smallskip
    \State $X_i \leftarrow \emptyset$ for $i \in [n]$; $P \leftarrow M$;
    \While{one of the update rules shown in Algorithm~\ref{Update Rules} is applicable}
    \Statex \textbf{Invariant:} \parbox[t]{0.8\textwidth}{$X$ is EFX and the envy-graph $G_X$ is acyclic}
    \State Let $U_\ell$ be an applicable update rule;
    \State $(X,P) \leftarrow U_\ell(X,P)$;   
    \State Decycle the envy-graph; 
    \EndWhile
  \end{algorithmic}
  \caption{Algorithm for Computing an EFX-Allocation}\label{main algorithm}
\end{algorithm}

In order to define our update rules, we need the concepts of {\em envy-graph} and the {\em most envious agent} for a bundle of goods. These were discussed in Section~\ref{sec:techniques} and we formally define them below. 

\begin{definition}
  The \emph{envy-graph} $G_X$ for an allocation $X = \langle X_1,X_2, \ldots ,X_n \rangle$ has the set of agents as vertices and there is a directed edge from agent $i$ to agent $j$ if and only if $v_i(X_i) < v_i(X_j)$. 
\end{definition}

The notion of envy-graph was introduced in~\cite{LiptonMMS04} and it is well-known that cycles can be removed from the envy-graph without destroying desirable properties (see Lemma \ref{decylcifying}). Thus we can maintain $G_X$ as a DAG. 

For each agent $s$, we define the \emph{reachability component} $C(s)$ as all agents reachable from $s$ in the envy-graph and the \emph{sources} of the envy-graph as the vertices with indegree zero. 

For ease of notation, we will use $B \setminus g$ and $B \cup g$ to denote $B \setminus \{g\}$ and $B \cup \{g\}$, respectively. 

\begin{lemma}
\label{decylcifying}
 Let $i_0 \rightarrow i_1 \rightarrow \cdots \rightarrow i_{k-1} \rightarrow i_0$ be a cycle in the envy-graph. Consider the allocation $X'$ where 
$X'_{i_{\ell}} = X_{i_{\ell + 1}}$ (indices are modulo $k$) for $\ell \in \{0,\ldots,k-1\}$ and $X'_j = X_j$ for $j \notin \{i_0,\ldots,i_{k-1}\}$. If $X$ is EFX, then $X'$ is also EFX. Moreover, $\phi(X') > \phi(X)$. 
\end{lemma}
\begin{proof} Consider any agent $i$. We have $v_i(X'_i) \ge v_i(X_i)$ with strict inequality if $i$ lies on the cycle. So $\sum_{i \in [n]}v_i(X'_i) > \sum_{i \in [n]}v_i(X_i)$. Thus $\phi(X') > \phi(X)$. 

Since $X'$ is just a permutation of $X$, for any agent $j$ there exists some agent $j'$ such that $X'_j=X_{j'}$. Therefore, since $X$ is EFX, for any good $g \in X_{j'}$ (or equivalently $X'_j$) we have $v_i(X'_j \setminus g) = v_i(X_{j'} \setminus g) \le v_i(X_i) \le v_i(X'_i)$. Thus $X'$ is also EFX. 
\end{proof}




Let $S \subseteq M$. Suppose there exists an agent that considers $S$ more valuable than her own bundle. Then we will call $S$
an {\em envied} set. The following definition formalizes the notion of the ``most envious agent''.

\begin{definition}
\label{mostenviousagent-definition}
Let $Z$ be an \emph{inclusion-wise minimal  envied subset} of $S$, i.e., 
(1)~there exists an agent $i$ such that $v_i(Z) > v_i(X_i)$ and (2)~for all $j \in [n]$ we have $v_j(X_j) \ge v_j(Z')$ for all
$Z' \subset Z$.
The agent $i$ (ties broken arbitrarily) will be called the {\em most envious agent} of $S$.
\end{definition}

%
%
%
%

We are now ready to present our three update rules $U_0$, $U_1$, and $U_2$, in Algorithm~\ref{Update Rules}.

\begin{algorithm*}[ht!]
  \begin{algorithmic}[1]
    \Function{$U_0$}{allocation $X$, pool $P$}
    \Statex \textbf{Precondition:} \parbox[t]{0.8\textwidth}{There is a good $g \in P$ and an agent $i$ such that allocating $g$ to $i$ results in an EFX allocation. }\smallskip
    \State Allocate $g$ to $i$, i.e., $X'_i \leftarrow X_i \cup g$, $P' \leftarrow P \setminus g$, and $X'_j = X_j$ for $j \not= i$.
    \State \textbf{return} $(X',P')$.
    \EndFunction\medskip

    \Function{$U_1$}{allocation $X$, pool $P$}
     \Statex \textbf{Precondition}: \parbox[t]{0.8\textwidth}{There is an agent $k$ such that $v_k(P) > v_k(X_k)$.}
     \State Let $i$ be the most envious agent of $P$; let $Z$ be an inclusion-wise minimal envied subset 

     of $P$
     such that $i$ satisfies $v_i(Z) > v_i(X_i)$.
     \State Set $X'_i = Z$ and $X'_j = X_j$ for $j \not= i$.
     \State  Set $P' = X_i \cup (P \setminus Z)$. 
     \State \textbf{return} $(X',P')$.
     \EndFunction\medskip

     \Function{$U_2$}{allocation $X$, pool $P$}
     \Statex \textbf{Precondition}: \parbox[t]{0.8\textwidth}{There is an $\ell \ge 1$, distinct goods $g_0$, $g_1$, \dots , $g_{\ell -1 }$ in $P$, distinct sources $s_0$, $s_1$, \dots, $s_{\ell-1}$ of $G_X$ and distinct agents $t_1$, $t_2$,\dots , $t_{\ell}$ such that $t_i \in C(s_i)$ and $t_{i+1}$ is the most envious agent of $X_{s_{i}} \cup g_{i}$ for $0 \le i \le \ell-1$ (indices are to be interpreted modulo $\ell$). }\smallskip 
     \State Let  $Z_{i}$ be an inclusion-wise minimal envied subset of $X_{s_{i}} \cup {g_{i}}$ such that 
     
     $v_{t_{i+1}}(Z_{i}) > v_{t_{i+1}}(X_{t_{i+1}})$ for $0 \le i \le \ell - 1$.
     \State Set $P' = (P \setminus \cup_{i=0}^{\ell}\{g_i\}\bigcup_{i = 0}^{\ell-1}((X_{s_i} \cup g_i) \setminus Z_i))$.
     \State Let $u_0^{i} \rightarrow \cdots \rightarrow u_{m_i}^{i}$ be the path of length $m_i$ from $s_i = u_0^{i}$ to $t_i= u_{m_i}^{i}$

     in $C(s_i)$ for $0 \le i \le \ell - 1$. 
     \State Set $X'_{u^{i}_k} = X_{u^{i}_{k+1}}$ for all $k \in \{0,\dots,m_i-1\}$ and all $i \in \{0,\ldots,\ell-1\}$. 
     \State Set $X'_{t_{i}} = Z_{i-1}$ for all $i \in \{1,\ldots,\ell\}$. 
     \State Set $X'_j = X_j$ for all other $j$.
     \State \textbf{return} $(X',P')$.
     \EndFunction
    \end{algorithmic}
  \caption{The Update Rules}
  \label{Update Rules}
\end{algorithm*}

\smallskip

\noindent{\bf Rule $\mathbf{U_0}:$} Rule $U_0$ is the easiest of the update rules. It is applicable whenever adding a good from the pool to some source of $G_X$ does not destroy the EFX-property (see Algorithm \ref{Update Rules}).

\begin{lemma}[Rule $U_0$]\label{Rule U0}\mbox{}
	\begin{itemize}
		\item[(a)] Rule $U_0$ returns an EFX allocation. An application of the rule does not decrease social welfare and decreases the size of the pool. 
		\item[(b)] If rule $U_0$ is not applicable then for any source $i$ of $G_X$ and good $g \in P$, there will be an agent $j \not= i$ such that $v_j(X_i \cup g) > v_j(X_j)$. In particular, an inclusion-wise minimal envied subset of $X_i \cup g$ has size at most $\abs{X_i}$.
	\end{itemize}
\end{lemma}
\begin{proof} 
	
	The first part of  a) follows directly from the precondition of the rule. The second part holds since the valuations are monotone and because $|P'| = |P| - 1$.
	
	The first sentence in part b) is obvious. We come to the second sentence. Since adding $g$ to $X_i$ destroys the EFX-property, there must be some $g' \in X_i \cup g$ and some $j \in [n]$ such that $v_j(X_i \cup g \setminus g') > v_j(X_j)$ for some $j \in [n]$. Thus an inclusion-wise minimal envied subset of $X_i \cup g$ has size at most $\abs{X_i}$.
\end{proof}

\smallskip

\noindent{\bf Rule $\mathbf{U_1}:$} Rule $U_1$ is applicable whenever there is an agent that values the pool higher than her current bundle (see Algorithm \ref{Update Rules}). 

\begin{lemma}[Rule $U_1$] Rule $U_1$ increases the social welfare and returns an EFX allocation. 
\end{lemma}
\begin{proof} Since there is an agent that values the pool higher than her own bundle, there is an inclusion-wise minimal envied
  subset $Z$ of $P$ and an agent $i$ such that $v_i(X_i) < v_i(Z)$.
  Let $X'$ be the allocation defined in Algorithm~\ref{Update Rules}, line~7. Then $v_i(X'_i) > v_i(X_i)$ and $v_j(X'_j) = v_j(X_j)$ for $j \not= i$. Thus $\phi(X') > \phi(X)$. It remains to show that the allocation $X'$ is EFX, i.e., for every pair of agents $j$ and $k$ and any good $g \in X'_k$, we have $v_j(X'_k \setminus g) \le v_j(X'_j)$.

  Since $X$ is EFX, this is obvious if neither $j$ nor $k$ is equal to $i$. If $j = i$, then $v_i(X'_i) > v_i(X_i) \geq v_i(X_k \setminus g) = v_i(X'_k \setminus g)$ for all $g \in X'_k$ (or equivalently $g \in X_k$). Finally, we consider $k=i$. Since $Z$ is an
  inclusion-wise minimal envied subset of $P$, we have $ v_j(X'_j)=v_j(X_j) \geq v_j(Z \setminus g) = v_j(X'_i \setminus g)$ for any $g \in Z$,
  where $v_j(X_j) \geq v_j(Z \setminus g)$ follows from the inclusion-wise minimality of $Z$ as an envied set. 
\end{proof}

\noindent{\bf Rule $\mathbf{U_2}:$} Rule $U_2$ is our most complex rule. It is applicable if for some $\ell \ge 1$, there are  distinct goods $g_0$, $g_1$, \dots , $g_{\ell -1 }$ in $P$, distinct sources $s_0$, $s_1$, \dots, $s_{\ell-1}$ of $G_X$ and distinct agents $t_1$, $t_2$,\dots , $t_{\ell}$ (indices are to be interpreted modulo $\ell$) such that for each $i$: (1)~$t_i$ is the most envious agent of $X_{s_{i-1}} \cup g_{i-1}$ and (2)~$t_i$ is reachable from $s_i$. We first show that rule $U_2$ is applicable if rule $U_0$ is not applicable and the pool contains at least $n$ goods.

\begin{figure}[h]
\centerline{\resizebox{0.38\textwidth}{!}{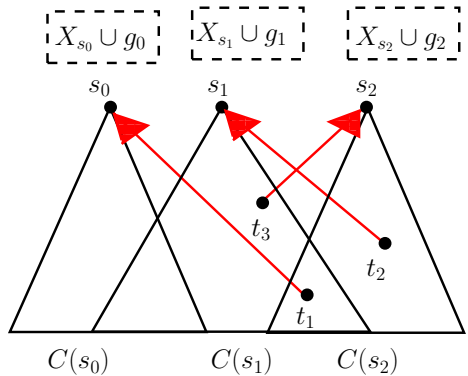}}
\caption{We have $t_i$ as the most envious agent of $X_{s_{i-1}} \cup g_{i-1}$. Moreover, $t_i \not\in C(s_0) \cup \cdots \cup C(s_{i-1})$ for $i = 1,2$ and $t_3 \in C(s_0) \cup \cdots \cup C(s_{2})$. Note that $j = 1$ is the largest index such that $t_3 \in C(s_j)$. The cycle is defined by $s_1$, $s_2$, $g_1$, $g_2$, $t_2$ and $t_3$. }
\label{Rule2}
\end{figure}

\begin{lemma}
\label{cycle-existence}
 If $\abs{P} \ge n$ and rule $U_0$ is not applicable then there is an $\ell \ge 1$, distinct goods $g_0,g_1, \dots ,g_{\ell-1}$ in $P$, distinct sources $s_0,s_1, \dots ,s_{\ell-1}$ of $G_X$, and distinct agents $t_1,t_1,\dots , t_{\ell}$ such that
 $t_i \in C(s_i)$ and $t_i$ is the most envious agent of $X_{s_{i-1}} \cup {g_{i-1}}$ for $i \in \{0,\ldots,\ell-1\}$ (indices are modulo $\ell$). 
\end{lemma}

\begin{proof} Since rule $U_0$ is not applicable, for every source $s$ of $G_X$ and every good $g \in P$, we have $v_j(S') > v_j(X_j)$ for some $S' \subset X_s \cup g$ and $j \in [n]$. Construct a sequence of triples $(s_i,g_i,t_{i+1})$, $i \ge 0$ defined as follows. Let $s_0$ be an arbitrary source of $G_X$ and $g_0$ be an arbitrary good in $P$. Assume we have defined $s_{i-1}$ and $g_{i-1}$. Let $t_{i}$ be the most envious agent of $X_{s_{i-1}} \cup {g_{i-1}}$.  If $t_i \not\in C(s_0) \cup \dots \cup C(s_{i-1})$, let $s_i$ be such that $t_i \in C(s_i)$. Also, let $g_i$ be a good in $P$ distinct from $g_0$ to $g_{i-1}$. 
If $t_{i} \in C(s_0) \cup \dots \cup C(s_{i-1})$ then stop the construction of the sequence and let $j$ be the maximum index such that $t_i \in C(s_j)$. Set $\ell = i - j$ and return $s_j,\dots,s_{i-1}$, $g_j,\dots,g_{i-1}$ and $t_{j+1},\dots,t_i$;\footnote{We took $j$ to be the maximum index such that $t_i \in C(s_j)$ so that the cycle defined by $s_j,\dots,s_{i-1}$, $g_j,\dots,g_{i-1}$ and $t_{j+1},\dots,t_i$ is simple.}  see Figure~\ref{Rule2} for an illustration.

  The construction is well-defined since $\abs{P} \ge n$ and hence we cannot run out of goods. The sources and goods are pairwise distinct by construction. The agents $t_1$ to $t_{i-1}$ are also distinct by construction. Finally, agent $t_{i}$ is distinct from any agent $t_{k}$ for $j < k <i$ since $t_i \in C(s_j)$ and by definition $t_{k} \notin (C(s_0) \cup \cdots \cup C(s_{k-1}))$ and so $t_{k} \notin C(s_j)$.
\end{proof}

For each $i$, let $u_0^{i} \rightarrow u_1^{i} \rightarrow \cdots \rightarrow u_{m_i}^{i}$ be the path of length $m_i$ from $s_i = u_0^{i}$ to $t_i = u_{m_i}^{i}$ in $C(s_i)$. Rule~$U_2$ assigns (i)~$X'_{u^{i}_k} = X_{u^{i}_{k+1}}$ for all $k \in \{0,\dots,m_i-1\}$ and all $i \in  \{0,\ldots,\ell-1\}$ and (ii)~$X'_{t_{i}} = Z_{i-1}$ for all $i \in \{1,\ldots,\ell\}$, where $Z_i$ is defined in Algorithm~\ref{Update Rules} (see line~12). For all other $j$, we have $X'_j = X_j$.

\begin{lemma}[Rule $U_2$] Rule $U_2$ increases social welfare and returns an EFX allocation. \end{lemma}

\noindent{\em Proof.}      We first observe that the valuations of the agents for their bundles have either increased or remained the same (since either the agents are left with their old bundles  or assigned bundles that they envied). In particular, the valuations of all the agents in 
$\bigcup_{i=0}^{\ell-1}\ \bigcup_{k=0}^{m_i} \{u^{i}_k\}$  are strictly larger, where 
the vertices $u^i_k$ are defined above. 
Thus $\phi(X')> \phi(X)$.
      
It remains to show that the allocation $X'$ is EFX, i.e., for every pair of agents $j$ and $k$ and any good $g \in X'_k$ we have $v_j(X'_k \setminus g) \le v_j(X'_j)$.  Let $T = \sset{t_1,t_2, \dots, t_{\ell}}$. For every agent $r \notin T$ we have $X'_{r} = X_{r'}$ for some $r'$. Now consider two cases depending on $k$:
     \begin{itemize}
         \item[--] \emph{$k \notin T$}$\colon$ Note that the agents' valuation for their current bundle (in $X'$) is at least as good as their valuation for their old bundle (in $X$). We have $v_j(X'_j) \geq v_j(X_j) \geq v_j(X_{k'} \setminus g) = v_j(X'_k \setminus g)$ for any $g \in X'_k$ (or equivalently, $g \in X_{k'}$).
         \item[--] \emph{$k \in T$}$\colon$ Let $k =t_i$ for some $i$. We have $v_j(X'_j) \geq v_j(X_j) \geq v_j(Z_{i-1} \setminus g)$ for any $g \in Z_{i-1}$ (by the inclusion-wise minimality of $Z_{i-1}$ as an envied set) and $v_j(Z_{i-1} \setminus g) = v_j(X'_{t_i} \setminus g) = v_j(X'_{k} \setminus g)$ for any $g \in X'_k$. \qed
         \end{itemize}

     We can now summarize. Let $V = \max_i v_i(M)$ be the maximum valuation of any agent and $\Delta = \min_i \min \{|v_i(T) - v_i(S)|: S,T \subseteq M \text{ and } v_i(S) \ne v_i(T)\}$ be the minimum difference between distinct valuations.
     Each application of rule $U_1$ or rule $U_2$ increases the social welfare by at least $\Delta$ and hence there can be no more than $n V/\Delta$ applications of these rules. Each application of rule $U_0$ decreases the size of the pool by one and hence there cannot be more than $m$ successive applications of this rule. We conclude that the number of iterations is at most $nmV/\Delta$. Thus we have shown the following theorem.

\begin{theorem}
\label{existence of EFX intro} 
For normalized and monotone valuations, there is always an allocation $X$ and a pool $P$ of unallocated goods such that 
\begin{itemize}
\item $X$ is EFX, 
\item $v_i(X_i) \geq v_i(P)$ for all agents $i$, and 
\item $|P|$ is less than the number of sources in the envy-graph; in particular, $|P| < n$.
\end{itemize}
Algorithm~\ref{main algorithm} determines such an allocation in at most $nmV/\Delta$ iterations.
\end{theorem}

\subsection{Finding an inclusion-wise minimal envied subset}
\label{minimum-envied-subset-finding}
We now describe how to implement a {\em minimal-envied-subset} oracle efficiently, i.e., given an envied set $S \subseteq M$,
we need to efficiently find an inclusion-wise minimal subset $Z$ of $S$ that is envied. Algorithm~\ref{minimal-valuable-subset}
finds such a set $Z \subseteq S$ for monotone valuations.

\begin{algorithm}[h]
  \begin{algorithmic}[1]
    \Statex {\textbf{Precondition}: \parbox[t]{0.8\textwidth}{There is an agent that values $S \subseteq M$ more than her own bundle.}}
    \Statex {\textbf{Postcondition}: \parbox[t]{0.8\textwidth}{$Z \subseteq S$ such that there is some agent $i$ with $v_i(Z) > v_i(X_i)$ and for all $j \in [n]$ we have $v_j(X_j) \ge v_j(Z')$ for all $Z' \subset Z$.}}
    \smallskip
    \State $Z = S$;
    \For {every agent $i$}
    \For {every good $g \in Z$}
    \If {$v_i(X_i) < v_i(Z\setminus g)$}
    \State $Z = Z \setminus g$; 
    \EndIf
    \EndFor
    \EndFor
    \State Return $Z$.
  \end{algorithmic}
  \caption{Algorithm for finding an inclusion-wise minimal envied subset}
  \label{minimal-valuable-subset}
\end{algorithm}

In Algorithm~\ref{minimal-valuable-subset}, if the ``if condition'' is not invoked for any agent $i$ and good $g \in Z$, then $Z$ is an inclusion-wise minimal envied subset. Suppose agent~$i$ envies $Z\setminus g$ for some $g \in Z$. Then we continue our check with $Z \setminus g$ as the new $Z$.  Since the new $Z$ is a subset of the old $Z$, observe that for each agent $i \in [n]$ and each good $g \in Z$, our algorithm needs to check exactly once if $v_i(X_i) \ge v_i(Z \setminus g)$ or not. Thus our algorithm makes $nm$ value queries. We can conclude the following theorem.

\begin{theorem}
\label{thm:pseudo-poly}
For normalized and monotone valuations, we can determine an allocation $X$ and a pool of unallocated goods $P$ that satisfies the 3 conditions in
Theorem~\ref{existence of EFX intro} using at most $n^2m^2V/{\Delta}$ value queries.
\end{theorem}

\subsection{An FPTAS to Determine an ``Almost'' Desired Allocation}
Algorithm~\ref{main algorithm} is pseudo-polynomial, since the increase in individual valuations of the agents when we perform the update rules could be very small. Suppose we just wanted an ``almost" EFX property, i.e., for every pair of agents $i$ and $j$, we are happy to ensure that $(1+\varepsilon)\cdot v_i(X_i) \geq v_i(X_j)$ and also $(1+\varepsilon)\cdot v_i(X_i) \geq v_i(P)$ for every $i$. Then we have an algorithm that runs in $\textup{poly}(n,m,\frac{1}{\varepsilon},\log V)$ time and finds a desired allocation.

\begin{theorem}
\label{thm:fptas}
For normalized and monotone valuations, given any $\varepsilon > 0$, using
$\textup{poly}(n,m,\frac{1}{\varepsilon},\log V)$ value queries,
we can determine an allocation $X = \langle X_1,X_2,\dots ,X_n \rangle $ and a pool of unallocated goods $P$ such that

  \begin{itemize}
  \item for any pair of agents $i$ and $j$, we have
    $(1+\varepsilon)\cdot v_i(X_i) \geq v_i(X_j \setminus g)$ for all $g \in X_j$,
      \item for any agent $i$, we have $(1+ \varepsilon)\cdot v_i(X_i) \geq v_i(P)$, and 
      \item $\lvert P \rvert < n$.
  \end{itemize}
\end{theorem}

The proof follows in a straightforward manner from the proof of Theorem~\ref{existence of EFX intro} in Section~\ref{sec:main}. The key idea is that the ``almost'' EFX property is violated if and only if $(1+\varepsilon)\cdot v_i(X_i) < v_i(X_j \setminus g)$ for some $i,j \in [n]$ or $(1+\varepsilon)\cdot v_i(X_i) < v_i(P)$ for some $i \in [n]$. So every time we apply update rules $U_1$ or $U_2$, there is a multiplicative improvement (by a factor of $1+\varepsilon$) in the valuation of some agents. Since these valuations are upper-bounded by $V$ we get a bound of $\textup{poly}(n,m,\log_{(1+\varepsilon)}V)$ on the number of iterations.    

\subsection{New Proof of a Result from \cite{TimPlaut18} for Identical Valuations}
\label{sec:appendix}
For agents with identical (general) valuations, it was shown by Plaut and Roughgarden~\cite{TimPlaut18} that an allocation that maximizes the minimum valuation, then maximizes the size of this bundle, then  maximizes the second minimum valuation, then maximizes the size of this bundle, and so on is EFX. We now show that Algorithm~\ref{main algorithm} gives another proof that when all the agents have identical valuations, a complete allocation that is EFX always exists. 

Recall that Algorithm~\ref{main algorithm} consists of applying 3 update rules: $U_0, U_1, U_2$ -- whichever of these is applicable. Moreover, if a certain precondition is satisfied (see Algorithm~\ref{Update Rules}), then rule $U_2$ is applicable. 

We will now show that when all the agents have identical valuations and rule $U_0$ is not applicable, then the precondition of rule $U_2$ is satisfied as long as there is some unallocated good. Let $X = \langle X_1,\ldots,X_n\rangle$ be the current allocation and $P = M \setminus \cup_{i=1}^n X_i$ be the set of unallocated goods in $X$. 

\begin{lemma}
\label{lemma:appendix}
 Let $s$ be any source vertex in the envy-graph $G_X$. If $\abs{P} \ge 1$ and rule $U_0$ is not applicable then $s$ can be chosen as the most envious agent of $X_s \cup g$ for any $g \in P$.
\end{lemma}
\begin{proof}
  Let $g \in P$ and $s$ be any source in the envy-graph $G_X$. Since rule $U_0$ is not applicable, there is an agent $t$ such that
  $v(X_t) < v(S')$ for some $S' \subset S = X_s \cup g$, where $v$ is the common valuation function of all agents.
  Let $Z \subseteq S$ be an inclusion-wise minimal envied subset of $S$ such that $v(X_t) < v(Z)$. Since $s$ is a source in $G_X$, we have $v(X_s) \le v(X_t)$. So  $v(X_s) <  v(Z)$; thus $s$ can be chosen as the most envious agent of $X_s \cup g$.
\end{proof}

Lemma~\ref{lemma:appendix} implies that while $P \ne \emptyset$, either rule $U_0$ or rule $U_2$ is applicable. Whenever we apply rule $U_2$, we add any good $g$ in $P$ to the bundle of a source $s$ in $G_X$ and determine an inclusion-wise minimal envied set $Z \subseteq X_s \cup g$. Note that $v(Z) > v(X_s)$. We then throw the goods in $(X_s \cup g) \setminus Z$ back into the pool $P$ and set $X_s = Z$.

This makes agent~$s$ strictly better-off and no agent is worse-off: thus we have made progress. So when Algorithm~\ref{main algorithm} 
terminates, we have an EFX allocation with $P = \emptyset$. Thus we have a complete allocation $X = \langle X_1,\ldots,X_n\rangle$ that is
EFX.

\paragraph{EFX for 2 agents.}
As observed in \cite{TimPlaut18}, an interesting consequence of the existence of EFX when all agents have identical valuations is the existence of EFX
for 2 agents with distinct valuations, say $v_1$ and $v_2$. First, an EFX allocation $\langle X_0, M\setminus X_0\rangle$ is obtained for 2 agents --
both with the same valuation function $v_1$. The {\em more valuable} set in $\{X_0, M\setminus X_0\}$ with respect to valuation function $v_2$ is assigned to agent~2 while agent~1 is assigned the other bundle. It is easy to see that this is an EFX allocation.

However finding such an EFX allocation may be hard, as shown by Plaut and Roughgarden~\cite{TimPlaut18}, and this is so even for two agents with identical submodular valuations. It was shown in \cite{TimPlaut18} that an exponential number of value queries may be required to determine an EFX allocation here. It follows from Theorem~\ref{thm:pseudo-poly} that an EFX allocation for two agents with general (possibly distinct) valuations can be determined using pseudo-polynomial many value queries. Moreover, even when agents have general valuations functions, we can determine a $(1-\varepsilon)$ EFX allocation with charity  (as in Theorem~\ref{thm:fptas})  with
$\textup{poly}(n,m,\frac{1}{\varepsilon},\log V)$ many value queries.

\section{Guarantees on Other Notions of Fairness}
\label{sec:additive-valuations}
In this section we assume that all agents have additive valuations.
We show that a minor variant of our algorithm finds an allocation with a good guarantee on Nash social welfare {\em and} groupwise maximin share (GMMS).

\medskip

\noindent{\bf Guarantee in Terms of Nash Social Welfare.}
We claimed in Section~\ref{sec:intro} that for additive valuations, it can also be ensured that for each $i$, we have $v_i(X_i) \ge \frac{1}{2}\cdot v_i(X^*_i)$ where $X^* = \langle X^*_1,\ldots,X^*_n\rangle$ is an  optimal {\em Nash social welfare} allocation and $X$ is the allocation in Theorem~\ref{existence of EFX intro}. This is easy to see from Algorithm~\ref{main algorithm}:
\begin{itemize}
\item[--]rather than initialize $X_i = \emptyset$, we will initialize $X_i$ to the bundle corresponding to the allocation determined by the algorithm in \cite{CaragiannisGravin19}.
\end{itemize}
  So we have $v_i(X_i) \ge \frac{1}{2}\cdot v_i(X^*_i)$, to begin with. As the algorithm progresses, our invariant is that $v_i(X_i)$ never decreases for any $i$. So if  $X' = \langle X'_1,\ldots,X'_n\rangle$ is the final allocation computed by our algorithm, then we have $v_i(X'_i) \ge \frac{1}{2}\cdot v_i(X^*_i)$ for all $i \in [n]$.

\begin{lemma}
\label{EFX with NSW}
 Given a set $N$ of agents with additive valuations and a set $M$ of goods, there exists an allocation $X = \langle X_1,\ldots,X_n\rangle$ and a pool $P$ of unallocated goods that satisfy all the conditions of Theorem~\ref{existence of EFX intro} and $v_i(X_i) \geq \tfrac{1}{2}v_i(X^{*}_i)$ for all $i \in N$, where $X^{*} = \langle X^*_1,\ldots,X^*_n\rangle$ is an  optimal Nash social welfare allocation. 
\end{lemma}

\subsection{An Approximate MMS Allocation for Large $\mathbf{|P|}$}
\label{sec:approx-MMS}
We now show that if $|P|$ (the number of unallocated goods in our allocation) is sufficiently large, then our EFX allocation $X$ has a very good MMS guarantee. Recall that our algorithm continues till $|P|$ is smaller than the number of sources in the envy-graph $G_X$ and recall that sources are {\em unenvied} agents. In particular, if $|P| = n-1$, then the number of sources in $G_X$ is $n$; so no agent envies another. That is, for each $i$, we have $v_i(X_i) \ge v_i(X_j)$ for all $j \in [n]$. Moreover, $v_i(X_i) \ge v_i(P)$. So we have 
\begin{eqnarray*}
  v_i(X_i) \ \geq \ \frac{v_i(M)}{n+1} \ & \geq & \ \left(1 + \frac{1}{n}\right)^{-1} \cdot\frac{v_i(M)}{n} \\
  & \ge & \ \left(1 + \frac{1}{n}\right)^{-1} \cdot \MMS_i(n,M),
\end{eqnarray*}  
where for every agent $i$, the inequality $\MMS_i(n,M) \leq v_i(M)/n$ holds for additive valuations. We formalize the above intuition in Theorem~\ref{EFX-MMS}. The following proposition from \cite{JGargMT19} will be useful. It states that if we exclude any set of agents and at most the same number of any goods from $N$ and $M$, respectively, the maximin share of any remaining agent can only increase.

\begin{proposition}[\cite{JGargMT19}]
   \label{MMS-reducability}
     Let $N$ be a set of $n$ agents with additive valuations and $M$ be a set of $m$ goods. If $N' \subseteq N$ and $M' \subseteq M$ are such that $|N \setminus N'| \geq |M \setminus M'|$ then for any agent $i \in N'$, we have $\MMS_i(n',M') \geq \MMS_i(n,M)$ where $n' = \abs{N'}$.
\end{proposition}

\begin{theorem}
 \label{EFX-MMS}
  Given a set $N$ of $n$ agents with additive valuations and a set $M$ of $m$ goods, there exists an allocation $X = \langle X_1,X_2, \dots, X_n \rangle$ and set $P$ of unallocated goods that satisfies:
\begin{itemize}
    \item the 3 conditions stated in Theorem~\ref{existence of EFX intro};
    \item $v_i(X_i) \geq \tfrac{1}{2}v_i(X^{*}_i)$ for all $i \in N$, where $X^{*}$ is an  optimal Nash social welfare allocation;
    \item $v_i(X_i) \geq \MMS_i(n,M)/\left(2 - \frac{k}{n}\right)$ for every $i \in N$, where $k = |P|$.
\end{itemize}
\end{theorem}

\noindent{\em Proof.}
Let $(X,P)$ be the allocation guaranteed by Lemma~\ref{EFX with NSW}.
Hence the first two conditions given in the theorem statement are satisfied by $(X,P)$.
So what we need to show now is that for any agent $i$, we have $v_i(X_i) \geq \MMS_i(n,M)/\left(2 - \frac{k}{n}\right)$. 

We fix some agent $i$ and let $N' \subseteq N$ be the set of agents consisting of all sources of $G_X$, agent~$i$ and all other agents $j$ with $\abs{X_j} \ge 2$. 
Let $M'$ be the set of goods allocated to the agents in $N'$. 
Observe that every agent in $N \setminus N'$ is allocated at most one good and so $\abs{N \setminus N'} \geq \abs{M \setminus (M' \cup P)}$.
By Proposition~\ref{MMS-reducability}, it holds that $\MMS_i(n',M' \cup P) \geq \MMS_i(n,M)$ where $n' = \abs{N'}$. Thus, it suffices to show that $v_i(X_i) \geq \MMS_i(n',M'\cup P)/\left(2 - \frac{k}{n}\right)$. 

\smallskip

Consider any agent $j \in N'$ with $\abs{X_j} \ge 2$.  
Because $X$ is EFX, it holds that $v_i(X_i) \geq v_i(X_j \setminus \left\{g\right\})$ for all $g \in X_j$. 
Since the valuations are additive, we have

%
%
%
%
%
\[v_i(X_i) \ \geq \ \left(1 - \frac{1}{\lvert X_j \rvert}\right)\cdot v_i(X_j) \ \geq \ \frac{1}{2}\cdot v_i(X_j).\] 

We know the following inequalities hold:
\begin{eqnarray}
v_i(X_i) &\geq&  v_i(P), \label{eq2} \\
v_i(X_i) &\geq&  v_i(X_j) \ \text{for all } j \ \text{that} \ \text{were sources in }G_X,  \label{eq3}\\
2v_i(X_i) &\geq & v_i(X_j)   \label{eq4} \ \text{for all other } j \in N'. \label{eq4}
\end{eqnarray}

Recall that the number of sources is at least $|P| + 1 = k+1$. Summing up all inequalities in (\ref{eq2})-(\ref{eq4}) and using the fact that $v_i$ is additive, we have $(2(n'-(k+1)) + k+2)\cdot v_i(X_i) \geq v_i(M' \cup P)$.

Hence we have 
\begin{align*}
    v_i(X_i) \ &\geq \ \frac{v_i(M'\cup P)}{2n'-k}\\
             \ &\geq \ \frac{v_i(M'\cup P)}{n'}\cdot\frac{n'}{2n'-k} \\
             \ &\geq \ \MMS_i(n',M'\cup P)\cdot\frac{n'}{2n'-k}  \ \ \text{(since $v_i$ is additive)} \\
             \ &= \ \MMS_i(n',M'\cup P)/\left(2 - \tfrac{k}{n'}\right)\\
             \ &\geq \ \MMS_i(n',M'\cup P)/\left(2 - \tfrac{k}{n}\right) \ \  \text{(since } n' \le n). \qed
\end{align*}

\subsection{An Improved Bound for Approximate-GMMS}
\label{sec:GMMS}
As mentioned in Section~\ref{sec:intro}, a new notion of fairness called {\em groupwise maximin share} (GMMS) was recently introduced by
Barman et al.~\cite{BarmanBMN18}. We formally define a GMMS allocation below.

\begin{definition}
Given a set $N$ of $n$ agents and a set $M$ of $m$ goods, an allocation $X = \langle X_1,X_2, \dots , X_n \rangle$ is $\alpha$-GMMS if for every $N' \subseteq N$, we have $v_i(X_i) \geq \alpha\cdot\MMS_i(n', \bigcup_{i \in N'}X_i)$ where $n' = \abs{N'}$. 
\end{definition}

Observe that a GMMS allocation is also an MMS allocation. Since MMS allocations do not always exist in a given instance~\cite{ProcacciaW14}, GMMS allocations also need not always exist.
Interestingly, $\tfrac{1}{2}$-GMMS allocations always exist~\cite{BarmanBMN18}. We now describe how to modify our allocation so that the
resulting allocation is $\tfrac{4}{7}$-GMMS.

Let $X = \langle X_1,\ldots,X_n\rangle$ be the allocation and $P$ be the pool of unallocated goods that satisfy the conditions of Lemma~\ref{EFX with NSW}. Without loss of generality, assume that agent~1 is a source in the envy-graph $G_X$. Define the complete allocation $Y = \langle Y_1,\ldots,Y_n\rangle$ as follows:
\begin{itemize}
\item[$\ast$] $Y_1 = X_1 \cup P$ and $Y_i=X_i$ for all $i \ne 1$.
\end{itemize}
Theorem~\ref{GMMS-guarantees} shows that $Y$ is our desired allocation.
The proof of Theorem~\ref{GMMS-guarantees} is similar to \cite[Proposition~3.4]{ABM18}. We also remark that one can use the proof of \cite[Proposition~3.4]{ABM18} to show that any EFX allocation is a $4/7$-GMMS. However note that $Y$ is not necessarily an EFX allocation. But it has sufficiently nice properties so that we can still show that it is $4/7$-GMMS.

\begin{theorem}
\label{GMMS-guarantees}
Given a set $N$ of $n$ agents with additive valuations and a set $M$ of $m$ goods, there exists a complete allocation $Y = \langle Y_1,Y_2,\dots, Y_n \rangle$ of $M$ such that 
\begin{itemize}
  \item $Y$ is $\tfrac{4}{7}$-GMMS.
  \item $v_i(Y_i) \geq \tfrac{1}{2} v_i(X^{*}_i)$ for all $i \in N$ where $X^*$ is the optimal Nash social welfare allocation.\footnote{In private communication we are aware that Jugal Garg and Setareh Taki have obtained independently related results. For \textbf{additive} valuations, they can show that there is an EFX-allocation after donating at most $n-1$ goods to charity. However, there is no bound on the value of the goods donated to charity. Thus they obtain a 4/7-GMMS-allocation after removing $n-1$ goods from the original set of goods.}
\end{itemize} 
\end{theorem}
\begin{proof}
Observe that the bound on Nash social welfare holds for allocation $X$ and thus for allocation $Y$ (since $v_i(Y_i) \ge v_i(X_i)$ for all $i \in [n]$). So what we need to show now is the guarantee on GMMS.
That is, we need to show that
for every $\widetilde{N} \subseteq N$ and all $i \in \widetilde{N}$, we have $v_i(Y_i) \geq \tfrac{4}{7} \MMS_i(\widetilde{n}, \widetilde{M})$ where $\widetilde{n} = \abs{\widetilde{N}}$ and $\widetilde{M} = \bigcup_{j \in \widetilde{N}}Y_j$.

Fix some $i \in \widetilde{N}$. Define $N'$ as the subset of $\widetilde{N}$ that contains $i$ and all agents that have been allocated at least two goods in $Y$, i.e., $j \in N'$ if and only if $j=i$ or $\abs{Y_j} \geq 2$. Let $M' = \bigcup_{j \in N'}Y_j$. 

Note that $Y$ allocates all goods of $\widetilde{M} \setminus M'$ to agents in $\widetilde{N} \setminus N'$. Since every agent in $\widetilde{N} \setminus N'$ has been allocated at most one good, we have $\abs{\widetilde{N} \setminus N'} \geq \abs{\widetilde{M} \setminus M'}$. Proposition~\ref{MMS-reducability} tells us that $\MMS_i(n',M') \geq \MMS_i(\widetilde{n},\widetilde{M})$ where $n' = \abs{N'}$. Thus it suffices to show $v_i(Y_i) \geq 4/7\cdot\MMS_i(n',M')$.

\smallskip

We now classify the goods in $M'$ as good or bad. A good is \emph{good} if it contained in a $Y_j$ of cardinality at least three or is contained in $Y_1$ or $Y_i$. All other goods are \emph{bad}, i.e., a bad good is contained in a bundle of cardinality two different from $Y_1$ and $Y_i$. A bundle in $Y$ containing good goods is good. 

We will next reduce the problem further. Let $x$ be the number of bad goods in $M'$. Since the bad goods are contained in bundles of cardinality two, the good goods in $M'$ come from $n' - x/2$ good bundles of $Y$. 
As long as $x > n'$, we will apply a reduction step. Each reduction step will reduce the number of bad goods in $M'$ by two, the number of agents by one, will not decrease the $\MMS_i$-value, and will leave the quantity $n' - x/2$ and set of good goods in $M'$ unchanged. 

Let $Z = \langle Z_1,Z_2, \dots Z_{n'} \rangle$  be an optimal MMS partition for agent $i$ of the set $M'$ of goods. 
If there are more than $n'$ bad goods in $M'$, there is a set $Z_k$ with at least {\em two} bad goods, say $g_1$ and $g_2$. We distribute the goods in $Z_k \setminus \{g_1,g_2\}$ arbitrarily among the other sets in $Z$.  So we have a partition of the set $M' \setminus \sset{g_1,g_2}$ of goods into $n'-1$ many bundles. The value for agent~$i$ of any remaining bundles did not decrease. We set $M'$ to  $M' \setminus \sset{g_1,g_2}$ and decrement $n'$. Note that we reduced the number of bad goods by two, the number of agents by one, did not decrease $\MMS_i(n',M')$ and the set of good goods in $M'$ still come from the $n' - x/2$ good bundles in $Y$. 
We keep repeating this reduction until $M'$ contains at most $n'$ bad goods. At this point, we have a set $M'$ of goods and an integer $n'$ with the following properties: 
\begin{compactenum}[(1)]
\item The number $x$ of bad goods in $M'$ is at most $n'$.
\item $\MMS_i(n',M') \geq \MMS_i(\widetilde{n},\wM)$, and 
\item The set of good goods in $M'$ has not changed. They come from $n' - x/2$ good bundles in $Y$. 
\end{compactenum}

\smallskip

We will next relate the value of good and bad goods to the value of $Y_i$. 
\begin{lemma} We have
  \begin{compactenum}[a)]
  \item For any bad good $g$, $v_i(g) \le v_i(Y_i)$.
  \item $v_i(Y_1) \le 2v_i(Y_i)$.
  \item If $j \not= 1$ and $Y_j$ is a good bundle then $v_i(Y_j) \le 3/2 \cdot v_i(Y_i)$.
  \end{compactenum}
\end{lemma}
\begin{proof}
\begin{compactenum}[a)]
\item Let $Y_j = \sset{g,g'}$ be the bundle containing $g$. Since $j \neq 1$ (by definition of a bad good), we have 
$v_i(Y_i) \geq v_i(X_i) \geq  v_i(X_j \setminus g') = v_i(Y_j \setminus g') = v_i(g)$. 
\item Since agent~1 is a source, $v_i(X_i) \geq v_i(X_1)$. By Theorem~\ref{existence of EFX intro}, $v_i(X_i) \geq v_i(P)$. Therefore, $v_i(Y_1) = v_i(X_1 \cup P) = v_i(X_1) + v_i(P) \leq v_i(X_i) + v_i(X_i) = 2v_i(X_i)=2v_i(Y_i)$.
\item 
Let $g \in Y_j$ be such that $v_i(g)$ is minimal. Then $v_i(Y_j - g) \le v_i(Y_i)$ and $v_i(g) \le v_i(Y_j)/\abs{Y_j}$. Thus 
\begin{eqnarray*} v_i(Y_i)  \geq \left(1- \frac{1}{\abs{Y_j}}\right)\cdot v_i(Y_j) & \geq & \left(1- \frac{1}{3 }\right)\cdot v_i(Y_j)\\  
                                                                                  & = & \frac{2}{3}\cdot v_i(Y_j). 
\end{eqnarray*}
\end{compactenum}
\end{proof}

Now we are ready to show the bound on GMMS. We have $x$ bad goods in $M'$. The good goods in $M'$ come from $n' - x/2$ good bundles. Also $x \le n'$.
The total value of the good goods for agent $i$ is at most \[\frac{3}{2}(n'-\tfrac{x}{2}-2)\cdot v_i(Y_i) + v_i(Y_i) + 2v_i(Y_i) = \frac{3}{2}(n' -\tfrac{x}{2})\cdot v_i(Y_i),\]
since there are at most $n' - x/2 - \abs{\sset{1,i}}$ good bundles different from $Y_1$ and $Y_i$: each of value at most 
$(3/2\cdot v_i(Y_i))$, and $Y_1$ has value at most $2v_i(Y_i)$. 
Also, the total value of the bad goods for agent $i$ is at most $x\cdot v_i(Y_i)$, since there are $x$ many bad goods and each bad good is worth at most $v_i(Y_i)$. Therefore, 
\begin{align*}
v_i(M') &= v_i(\text{bad goods in $M'$}) + v_i(\text{good goods in $M'$})\\
       &\leq x\cdot v_i(Y_i) + \frac{3}{2}(n'-\frac{x}{2})\cdot v_i(Y_i) \\
       &=\frac{6n'+x}{4}\cdot v_i(Y_i)
       \leq \frac{7n'}{4}\cdot v_i(Y_i). 
\end{align*} 
\end{proof}

\section{Conclusions and Open Problems}
We studied the existence of EFX allocations when agents have general valuations. We showed that we can ensure such an allocation always exists when we donate a small number of goods that nobody envies to charity. The major open problem here is whether complete EFX allocations always exist. 
Plaut and Roughgarden~\cite{TimPlaut18} remarked that an instance with {\em no} complete EFX allocation may be easier to find in the setting of general valuations.  Our result on ``almost-EFX'' allocations for general valuations allows one to hope that complete EFX allocations always exist, at least for more structured valuations such as additive.


We also showed that we get guarantees in terms of other notions of fairness when agents have additive valuations. To the best of our knowledge, allocations with good guarantees (i.e., constant factor approximation) on Nash social welfare and MMS (as well as GMMS) were not known prior to our work. It would also be interesting to investigate whether these guarantees can be improved or if instances can be constructed where our guarantees are tight. We believe that our work is just the beginning towards determining an allocation that gives good guarantees with respect to several notions of fairness: an allocation that is \emph{universally} fair. 

Very recently, Amanatidis et al.\ \cite{ANM19} have announced an allocation that has good approximation guarantees simultaneously with respect to four notions of fairness when valuation functions are additive: in particular, their allocation is $(\phi-1)$-EFX and $2/(\phi+2)$-GMMS, where $\phi \approx 1.618$ is the golden ratio. Moreover, a fine-tuned version of their algorithm also achieves $4/7$-GMMS which matches our result (see Theorem~\ref{GMMS-guarantees}). They also show that for additive valuations, when $m$ is at most $n+2$, where $m$ is the number of goods and $n$ is the number of agents, GMMS (and hence, EFX) allocations always exist.

\bigskip

\noindent{\bf Acknowledgments.} We thank the reviewers of the conference version of our paper in SODA~2020 for their helpful comments.


\begin{thebibliography}{10}






\bibitem{ABM18}
Georgios Amanatidis, Georgios Birmpas, and Vangelis Markakis.
\newblock Comparing approximate relaxations of envy-freeness.
\newblock In {\em Proceedings of the Twenty-Seventh International Joint
  Conference on Artificial Intelligence, (IJCAI)}, pages 42--48, 2018.

\bibitem{AMNS17}
Georgios Amanatidis, Evangelos Markakis, Afshin Nikzad, and Amin Saberi.
\newblock Approximation algorithms for computing maximim share allocations.
\newblock {\em ACM Transactions on Algorithms}, 13(4):52:1--52:28, 2017.

\bibitem{ANM19}
Georgios Amanatidis, Apostolos Ntokos, and Evangelos Markakis.
\newblock Multiple Birds with One Stone: Beating $1/2$ for EFX and GMMS via Envy Cycle Elimination.
\newblock \url{https://arxiv.org/pdf/1909.07650.pdf}

\bibitem{AOSS17}
Nima Anari, Shayan~Oveis Gharan, Amin Saberi, and Mohit Singh.
\newblock {Nash} social welfare, matrix permanent, and stable polynomials.
\newblock In {\em Proceedings of the 8th Innovations in Theoretical Computer
  Science (ITCS)}, pages 36:1--12, 2017.

\bibitem{AMOV18}
Nima Anari, Tung Mai, Shayan~Oveis Gharan, and Vijay~V Vazirani.
\newblock {Nash} social welfare for indivisible items under separable
  piecewise-linear concave utilities.
\newblock In {\em Proceedings of the 29th Annual ACM-SIAM Symposium on Discrete
  Algorithms (SODA)}, pages 2274--2290, 2018.

\bibitem{BarmanBMN18}
Siddharth Barman, Arpita Biswas, Sanath Kumar~Krishna Murthy, and Yadati
  Narahari.
\newblock Groupwise maximin fair allocation of indivisible goods.
\newblock In {\em {AAAI}}, pages 917--924. {AAAI} Press, 2018.

\bibitem{BK17}
Siddharth Barman and Sanath~Kumar Krishnamurthy.
\newblock Approximation algorithms for maximin fair division.
\newblock In {\em Proceedings of the 18th ACM Conference on Economics and
  Computation (EC)}, pages 647--664, 2017.

\bibitem{BKV18}
Siddharth Barman, Sanath~Kumar Krishnamurthy, and Rohit Vaish.
\newblock Finding fair and efficient allocations.
\newblock In {\em Proceedings of the 19th ACM Conference on Economics and
  Computation (EC)}, pages 557--574, 2018.

\bibitem{BCFIMPVZ19}
Vittorio Bil\'{o}, Ioannis Caragiannis, Michele Flammini, Ayumi Igarashi,
  Gianpiero Monaco, Dominik Peters, Cosimo Vinci, and William~S. Zwicker.
\newblock Almost envy-free allocations with connected bundles.
\newblock In {\em Proceedings of the 9th Innovations in Theoretical Computer
  Science (ITCS)}, pages 305--322. {LIPIcs}, 2018.

\bibitem{BL16}
Sylvain Bouveret and Michel Lema\^itre.
\newblock Characterizing conflicts in fair division of indivisible goods using
  a scale of criteria.
\newblock In {\em Autonomous Agents and Multi-Agent Systems (AAMAS) 30, 2},
  pages 259--290, 2016.

\bibitem{budish2011combinatorial}
Eric Budish.
\newblock The combinatorial assignment problem: Approximate competitive
  equilibrium from equal incomes.
\newblock {\em Journal of Political Economy}, 119(6):1061--1103, 2011.

\bibitem{CaragiannisGravin19}
Ioannis Caragiannis, Nick Gravin, and Xin Huang.
\newblock Envy-freeness up to any item with high nash welfare: The virtue of
  donating items.
\newblock In {\em Proceedings of the 20th ACM Conference on Economics and
  Computation (EC)}, pages 527--545. {ACM}, 2019.

\bibitem{CaragiannisKMP016}
Ioannis Caragiannis, David Kurokawa, Herv{\'{e}} Moulin, Ariel~D. Procaccia,
  Nisarg Shah, and Junxing Wang.
\newblock The unreasonable fairness of maximum {Nash} welfare.
\newblock In {\em Proceedings of the 17th ACM Conference on Economics and
  Computation (EC)}, pages 305--322, 2016.

\bibitem{ChaudhuryCG0HM18}
Bhaskar~Ray Chaudhury, Yun~Kuen Cheung, Jugal Garg, Naveen Garg, Martin Hoefer,
  and Kurt Mehlhorn.
\newblock On fair division for indivisible items.
\newblock In {\em Proceedings of the 38th IARCS Annual Conference on
  Foundations of Software Technology and Theoretical Computer Science
  (FSTTCS)}, pages 25:1--25:17, 2018.

\bibitem{CGM20}
Bhaskar~Ray Chaudhury, Jugal Garg, and Kurt Mehlhorn.
\newblock EFX Exists for Three Agents.
\newblock \url{https://arxiv.org/pdf/2002.05119.pdf}

\bibitem{CDGJMVY17}
Richard Cole, Nikhil~R Devanur, Vasilis Gkatzelis, Kamal Jain, Tung Mai,
  Vijay~V Vazirani, and Sadra Yazdanbod.
\newblock Convex program duality, {F}isher markets, and {Nash} social welfare.
\newblock In {\em Proceedings of the 18th ACM Conference on Economics and
  Computation (EC)}, pages 459--460, 2017.


\bibitem{CKMS20}
Bhaskar Ray Chaudhury, Telikepalli Kavitha, Kurt Mehlhorn and Alkmini Sgouritsa
\newblock A Little Charity Guarantees Almost Envy-Freeness
\newblock In {\em Proceedings of the 31st ACM-SIAM Symposium on Discrete
	Algorithms (SODA)} pages 2658--2672, 2020.



\bibitem{CG15}
Richard Cole and Vasilis Gkatzelis.
\newblock Approximating the {Nash} social welfare with indivisible items.
\newblock In {\em Proceedings of the 47th ACM Symposium on Theory of Computing
  (STOC)}, pages 371--380, 2015.

\bibitem{GHM18}
Jugal Garg, Martin Hoefer, and Kurt Mehlhorn.
\newblock Approximating the {Nash} social welfare with budget-additive
  valuations.
\newblock In {\em Proceedings of the 29th Annual ACM-SIAM Symposium on Discrete
  Algorithms (SODA)}, pages 2326--2340, 2018.

\bibitem{JGargMT19}
Jugal Garg, Peter McGlaughlin, and Setareh Taki.
\newblock Approximating maximin share allocations.
\newblock In {\em Proceedings of the 2nd Symposium on Simplicity in Algorithms
  (SOSA)}, volume~69, pages 20:1--20:11. Schloss Dagstuhl - Leibniz-Zentrum
  fuer Informatik, 2019.

\bibitem{GargTaki19}
Jugal Garg and Setareh Taki.
\newblock An improved approximation algorithm for maximin shares.
\newblock {\em CoRR}, abs/1903.00029, 2019.

\bibitem{GhodsiHSSY18}
Mohammad Ghodsi, Mohammad~Taghi Hajiaghayi, Masoud Seddighin, Saeed Seddighin,
  and Hadi Yami.
\newblock Fair allocation of indivisible goods: Improvements and
  generalizations.
\newblock In {\em Proceedings of the 2018 {ACM} Conference on Economics and
  Computation (EC)}, pages 539--556, 2018.

\bibitem{GP14}
Jonathan~R. Goldman and Ariel~D. Procaccia.
\newblock Spliddit: unleashing fair division algorithms.
\newblock In {\em {SIGecom Exchanges 13(2)}}, pages 41--46, 2014.


\bibitem{KPW18}
David Kurokawa, Ariel~D. Procaccia, and Junxing Wang.
\newblock Fair enough: Guaranteeing approximate maximin shares.
\newblock {\em Journal of ACM}, 65(2):8:1--27, 2018.


\bibitem{LiptonMMS04}
Richard~J. Lipton, Evangelos Markakis, Elchanan Mossel, and Amin Saberi.
\newblock On approximately fair allocations of indivisible goods.
\newblock In {\em Proceedings of the 5th ACM Conference on Electronic Commerce
  (EC)}, pages 125--131, 2004.

\bibitem{TimPlaut18}
Benjamin Plaut and Tim Roughgarden.
\newblock Almost envy-freeness with general valuations.
\newblock In {\em Proceedings of the 29th Annual ACM-SIAM Symposium on Discrete
  Algorithms (SODA)}, pages 2584--2603, 2018.

\bibitem{ProcacciaW14}
Ariel~D. Procaccia and Junxing Wang.
\newblock Fair enough: guaranteeing approximate maximin shares.
\newblock In {\em Proceedings of the 15th ACM Conference on Economics and
  Computation (EC)}, pages 675--692, 2014.

\bibitem{Steinhaus48}
Hugo Steinhaus.
\newblock The problem of fair division.
\newblock {\em Econometrica}, 16(1):101--104, 1948.

\end{thebibliography}
\end{document}